\documentclass[11pt]{article}

\usepackage[margin=1in]{geometry}
\usepackage{amsmath,amssymb,amsthm,mathtools}
\usepackage{comment}
\usepackage{enumitem}
\usepackage{microtype}
\usepackage{xurl}
\usepackage[hidelinks]{hyperref}
\usepackage[nameinlink,noabbrev]{cleveref}
\crefname{lemma}{Lemma}{Lemmas}
\crefname{theorem}{Theorem}{Theorems}
\crefname{proposition}{Proposition}{Propositions}
\crefname{corollary}{Corollary}{Corollaries}
\crefname{definition}{Definition}{Definitions}
\crefname{remark}{Remark}{Remarks}

\newtheorem{theorem}{Theorem}[section]
\newtheorem{lemma}[theorem]{Lemma}
\newtheorem{proposition}[theorem]{Proposition}
\newtheorem{corollary}[theorem]{Corollary}

\theoremstyle{definition}
\newtheorem{definition}[theorem]{Definition}

\theoremstyle{remark}
\newtheorem{remark}[theorem]{Remark}

\newcommand{\Ls}{L_s}
\newcommand{\Tr}{\operatorname{Tr}}
\newcommand{\diam}{\operatorname{diam}}
\newcommand{\rank}{\operatorname{rank}}
\newcommand{\dist}{\operatorname{dist}}
\newcommand{\Lip}{\operatorname{Lip}}

\newcommand{\cE}{\mathcal{E}}

\newcommand{\elltwo}{\ell_2}
\newcommand{\Ds}{D_s}

\title{Euclidean Embeddings of Acyclic Hypergraph Line Metrics}
\author{Nicholas Popescu}
\date{}

\begin{document}
\maketitle

\begin{abstract}
We study Euclidean embeddings of higher-order line metrics arising from hypergraphs. For a fixed threshold \(s\geq 1\), the \(s\)-line graph \(L_s(H)\) has one vertex for each hyperedge of \(H\), with two hyperedges adjacent when they intersect in at least \(s\) vertices; its shortest-path metric is the native \(s\)-line metric. In the unrestricted setting, every finite graph occurs as the \(s\)-line graph of some hypergraph, so the worst-case Euclidean distortion is \(\Theta(\log n)\). We show that \(\alpha\)-acyclicity yields a sharp improvement to \(\Theta(\sqrt{\log\log n})\). Although the canonical join tree of an \(\alpha\)-acyclic hypergraph may have arbitrarily large stretch relative to the native metric, we prove that for every \(s\), the \(s\)-line graphs of \(\alpha\)-acyclic hypergraphs are exactly the dually chordal graphs. Using their constant-factor spanning-tree approximations together with classical tree-embedding bounds, we obtain an \(O(\sqrt{\log\log n})\) Euclidean embedding bound. We then show this is asymptotically tight by realizing every finite tree as the \(s\)-line graph of an \(\alpha\)-acyclic hypergraph. Thus, the worst-case distortion is \(\Theta(\sqrt{\log\log n})\), including for hypergraphs of rank at most \(3s+1\).
\end{abstract}

\section{Introduction}
\label{sec:introduction}

A fundamental question in the geometry of finite metric spaces is how faithfully a metric can be represented in a more structured space. A low-distortion embedding replaces the graph metric by distances in a normed space while controlling the multiplicative change in every
pairwise distance. By approximately preserving distances in a more structured space, metric embeddings allow algorithms to exploit geometric properties that may be difficult to access in the original metric, providing a powerful tool for approximation and optimization \cite{linial1995geometry}. Bourgain's theorem established that every $n$-point metric space admits an embedding into Hilbert space with distortion $O(\log n)$ \cite{bourgain1985lipschitz}. As this bound is asymptotically optimal for general metrics, a universal improvement is impossible without additional structure in the source metric \cite{linial1995geometry}. For trees, $n$-vertex tree metrics admit Euclidean embeddings with distortion $O(\sqrt{\log\log n})$, with the corresponding lower bound attained by complete binary trees \cite{matousek1999embedding, bourgain1986superreflexivity,linial2003binary}. We investigate the Euclidean embedding behavior of hypergraph metrics, and whether structural restrictions on hypergraphs yield bounds analogous to those for structured graph families.

Hypergraphs provide a natural language for systems with higher-order interactions. Unlike ordinary graphs, hypergraphs do not come with a unique canonical notion of distance. Our objects of interest are the hyperedges, and a natural way to find distance between two hyperedges is by how strongly they overlap and by whether one can move between them through a sequence of overlapping hyperedges. For a fixed \(s\ge1\), we say that two hyperedges are \(s\)-adjacent when they share at least \(s\) vertices. The graph whose vertices are the hyperedges and whose edges record this relation is the \(s\)-line graph \(L_s(H)\). Its shortest-path metric therefore measures how many steps across edges with enough overlap are needed to pass from one hyperedge to another.

We first show that connected native $s$-line metrics of unrestricted hypergraphs are exactly finite connected graph metrics. Thus, despite their greater combinatorial generality, hypergraphs introduce no additional worst-case Euclidean distortion under the native $s$-line metric. This motivates studying whether structural restrictions yield sharper bounds.

We focus on $\alpha$-acyclic hypergraphs, which arose in database theory \cite{beeri1983desirability,fagin1983degrees}. An $\alpha$-acyclic hypergraph admits a join tree, a tree on the hyperedges in which all hyperedges containing any fixed vertex form a connected subtree. This structure is the basis of algorithms for acyclic database schemes and
conjunctive queries \cite{gottlob2001hypertree}. Thus, $\alpha$-acyclic hypergraphs form a natural tree analogue of ordinary graphs. Since their hyperedges can be organized into a tree that preserves the connectivity of shared vertices, this structure suggests adapting methods for tree metrics to $\alpha$-acyclic hypergraphs. In particular, we ask whether the native $s$-line metric
is well approximated by the join tree.

The first part of this paper shows that this intuition is only partly correct. For every connected component $C$ of $\Ls(H)$, the restriction of any join tree to $C$ is indeed a spanning tree of the native graph. Moreover, we identify its exact multiplicative stretch in terms of the diameters of $s$-traces in the join tree. However, a join tree need not be a good metric approximation. Specifically, there are $\alpha$-acyclic hypergraphs for which the native $s$-line graph is a clique while every join tree is a path. Thus, tree representation can incur stretch linear in the number of hyperedges even when the native metric itself is trivial. Instead, we show that the native graph belongs to a well-studied graph class with a stronger metric property. Brandst\"adt, Dragan, Chepoi, and Voloshin proved that the ordinary line graphs of \(\alpha\)-acyclic hypergraphs are exactly the dually chordal graphs \cite{brandstadt1998dually}. We extend this characterization to
line graphs of arbitrary order.

\begin{theorem}
\label{thm:intro-characterization}
Fix \(s\ge1\).  A finite graph \(G\) is dually chordal if and only if there
exists an \(\alpha\)-acyclic hypergraph \(H\) such that
\[
    G\cong \Ls(H).
\]
\end{theorem}

Further, every connected dually chordal graph admits a spanning tree with multiplicative stretch at most \(4\) \cite{brandstadt1999distance}.  Combining this with Matou\v{s}ek's \(O(\sqrt{\log\log n})\) Euclidean embedding bound for \(n\)-point tree metrics \cite{matousek1999embedding} yields the corresponding upper bound for native \(s\)-line metrics of \(\alpha\)-acyclic hypergraphs.

Our main result shows that this upper bound is asymptotically tight.

\begin{theorem}
\label{thm:intro-main}
Fix \(s\ge1\), and let
\[
  \Ds(n)
  :=
  \sup\left\{
    c_2(C,d_s):
    \begin{array}{l}
      H\text{ is }\alpha\text{-acyclic},\\
      C\text{ is a connected component of }\Ls(H),\\
      |C|=n
    \end{array}
  \right\}.
\]
Then
\[
    \Ds(n)
    =
    \Theta\!\left(\sqrt{\log\log n}\right).
\]
Further, the lower bound already holds for hypergraphs of rank at most
\(3s+1\).
\end{theorem}

To prove the matching lower bound, we show that every finite tree can itself occur as the \(s\) line graph of an \(\alpha\)-acyclic hypergraph. Specifically, every finite tree \(Q\) occurs as the \(s\)-line graph of an \(\alpha\)-acyclic hypergraph \(H_{Q,s}\), while the rank of the construction
is controlled by
\[
    \rank(H_{Q,s})\le 1+s\Delta(Q).
\]
Applying this construction to complete binary trees transfers the classical Hilbert-space lower bound \(\Omega(\sqrt{\log\log n})\) \cite{bourgain1986superreflexivity,linial2003binary} to the class of native \(s\)-line metrics.  Since complete binary trees have maximum degree at most \(3\), the same construction gives the rank bound \(3s+1\).

\paragraph{Organization.}
Section~\ref{sec:preliminaries} introduces the native \(s\)-line metric, explains why we work on hyperedges rather than on vertices, and shows that unrestricted hypergraphs recover all finite graph metrics.  Section~\ref{sec:join-tree-comparison} compares the native metric with a join tree and characterizes the exact stretch of the canonical tree. Section~\ref{sec:canonical-limitations} gives examples showing that this stretch can be arbitrarily large.  We then characterize higher-order line graphs of \(\alpha\)-acyclic hypergraphs in terms of dually chordal graphs and derive the Euclidean upper bound.  Finally, Section~\ref{sec:tightness} proves the matching lower bound through explicit bounded-rank realizations of tree metrics.

\section{Preliminaries}
\label{sec:preliminaries}

In the following work, all hypergraphs and graphs are finite. We write a hypergraph as \(H=(V,E)\), where \(E\subseteq 2^V\) is a set of distinct hyperedges. The rank of \(H\) is
\[
    \rank(H):=\max_{e\in E}|e|.
\]
We fix an integer \(s\ge1\) throughout the paper.

\subsection{Higher-order line graphs and the native metric}
\label{subsec:s-line-metric}

The ordinary line graph of a hypergraph has one vertex for each hyperedge and joins two hyperedges whenever they intersect. To study different amounts of overlap, we study the \(s\)-line graph \cite{liu2022highorder,preti2024hyper}.

\begin{definition}[\(s\)-line graph]
\label{def:s-line}
The \emph{\(s\)-line graph} \(\Ls(H)\) has vertex set \(E\), with distinct
hyperedges \(e,f\in E\) adjacent if and only if
\[
    |e\cap f|\ge s.
\]
\end{definition}

The graph \(\Ls(H)\) is not necessarily connected. In particular, every hyperedge of size less than \(s\) is isolated. For a connected component \(C\) of \(\Ls(H)\), write
\[
    G_C:=\Ls(H)[C]
\]
and define
\[
    d_s(e,f):=d_{G_C}(e,f),
    \qquad e,f\in C.
\]
Thus \(d_s(e,f)\) is the minimum number of \(s\)-adjacency steps needed to move from \(e\) to \(f\). We call \(d_s\) the \emph{native \(s\)-line metric}. Metric statements throughout the paper are made componentwise.

\subsection{Why a metric on hyperedges?}
\label{subsec:why-hyperedges}

Higher-order connectivity also gives rise to natural notions of distance between vertices of a hypergraph. For distinct \(u,v\in V\), we define
\[
    \delta_s(u,v)
    :=
    \min_{\substack{e\ni u\\f\ni v}}
    \bigl(d_s(e,f)+1\bigr),
\]
with \(\delta_s(u,u)=0\), whenever the relevant hyperedges lie in a common \(s\)-connected component. Analogous vertex-to-hyperedge distances may be defined in the same way \cite{preti2024hyper}.

The hyperedge distance \(d_s\) is the ordinary shortest path metric of \(\Ls(H)\), and thus is a metric on each connected component. On the other hand, for \(s>1\), the corresponding vertex distances need not satisfy the triangle inequality and are in general only semimetrics \cite{preti2024hyper}. Since our goal is to study bi-Lipschitz embeddings into Euclidean space, we take the hyperedges as the points of the metric space.

\subsection{Euclidean distortion}
\label{subsec:distortion}

We use the standard bi-Lipschitz notion of metric distortion \cite{linial1995geometry}. If \((X,d_X)\) is a finite metric space and \(F:X\to\elltwo\) is injective, define
\[
    \Lip(F)
    :=
    \max_{x\neq y}
    \frac{\|F(x)-F(y)\|_2}{d_X(x,y)}
\]
and
\[
    \Lip(F^{-1})
    :=
    \max_{x\neq y}
    \frac{d_X(x,y)}{\|F(x)-F(y)\|_2}.
\]
The distortion of \(F\) is
\[
    \dist(F):=\Lip(F)\Lip(F^{-1}),
\]
and the least Euclidean distortion of \(X\) is
\[
    c_2(X):=
    \inf_{F:X\hookrightarrow\elltwo}\dist(F).
\]
For a graph \(G\), \(c_2(G)\) denotes the Euclidean distortion of its shortest-path metric.

Since distortion is invariant under rescaling, an embedding has distortion at most \(D\) if and only if it can be rescaled so that
\[
    d_X(x,y)
    \le
    \|F(x)-F(y)\|_2
    \le
    D\,d_X(x,y)
\]
for all \(x,y\in X\).

\subsection{Unrestricted hypergraphs recover arbitrary graph metrics}
\label{subsec:unrestricted}

Before imposing acyclicity, we observe that \(s\)-line metrics of arbitrary hypergraphs behave like graph metrics.

\begin{proposition}[Graph realization]
\label{prop:general-graph-realization}
For every integer \(s\ge1\) and every finite simple graph \(G=(U,F)\), there exists a hypergraph \(H_G^{(s)}\) such that
\[
    \Ls\!\left(H_G^{(s)}\right)\cong G.
\]
Therefore, if \(G\) is connected, the native \(s\)-line metric of \(H_G^{(s)}\) is isometric to the shortest-path metric of \(G\).
\end{proposition}

\begin{proof}
For every edge \(uv\in F\), let \(B_{uv}\) be a set of exactly \(s\) fresh vertices, with all blocks \(B_{uv}\) pairwise disjoint. For each \(u\in U\), introduce a private vertex \(p_u\), distinct from all other private vertices and from all vertices in the blocks, and define a hyperedge
\[
    E_u
    :=
    \{p_u\}\cup\bigcup_{uv\in F} B_{uv}.
\]
Let
\[
    H_G^{(s)}
    :=
    \left(
        \{p_u:u\in U\}\cup\bigcup_{uv\in F}B_{uv},
        \{E_u:u\in U\}
    \right).
\]
The private vertices ensure that the hyperedges \(E_u\) are pairwise distinct, including when \(G\) has isolated vertices.

For distinct \(u,v\in U\), if \(uv\in F\), then
\[
    E_u\cap E_v=B_{uv},
\]
and thus
\[
    |E_u\cap E_v|=s.
\]
If \(uv\notin F\), then no block belongs to both \(E_u\) and \(E_v\), so
\[
    E_u\cap E_v=\varnothing.
\]
Therefore
\[
    |E_u\cap E_v|\ge s
    \iff
    uv\in F.
\]
The correspondence \(u\mapsto E_u\) is thus a graph isomorphism
\[
    G\cong \Ls\!\left(H_G^{(s)}\right).
\]
When \(G\) is connected, graph isomorphism preserves shortest-path distance, so the two metric spaces are isometric.
\end{proof}

\begin{corollary}
\label{cor:all-graph-metrics}
For every fixed \(s\ge1\), the class of connected native \(s\)-line metrics of arbitrary hypergraphs is exactly the class of finite connected graph metrics.
\end{corollary}

\begin{proof}
Every connected native \(s\)-line metric is, by definition, the shortest-path metric of a connected graph, namely a connected component of \(\Ls(H)\). Therefore, by Proposition~\ref{prop:general-graph-realization}, every finite connected graph \(G\) is isomorphic to \(\Ls(H)\) for some hypergraph \(H\). So every finite connected graph metric occurs as a native \(s\)-line metric.
\end{proof}

In particular, unrestricted hypergraphs admit no better worst-case Euclidean embedding guarantee than arbitrary graph metrics. There exist \(n\)-vertex graph metrics requiring Euclidean distortion
\(\Omega(\log n)\) \cite{linial1995geometry}, while Bourgain's theorem gives an \(O(\log n)\) Euclidean embedding for every \(n\)-point metric space. 
Thus,
\[
    \sup_{\substack{H,\ C\text{ a component of }\Ls(H)\\|C|=n}}
    c_2(C,d_s)
    =
    \Theta(\log n).
\]
Thus any improvement over the general metric bound must come from adding structure to the hypergraph.

\subsection{Alpha-acyclicity and join trees}
\label{subsec:alpha-acyclicity}

There are several inequivalent notions of hypergraph acyclicity. We work with \(\alpha\)-acyclicity, the notion arising in relational database theory \cite{beeri1983desirability,fagin1983degrees}. \(\alpha\)-acyclic hypergraphs form a natural tree-like class of hypergraphs since hyperedges admit a join tree representation.

\begin{definition}[Join tree]
\label{def:join-tree}
A \emph{join tree} for \(H=(V,E)\) is a tree \(J\) whose vertex set is \(E\)
and such that, for every \(x\in V\), the set
\[
    E_x:=\{e\in E:x\in e\}
\]
induces a connected subtree of \(J\). Equivalently, whenever
\(x\in e\cap f\), every hyperedge on the unique \(J\)-path from \(e\) to
\(f\) also contains \(x\).
\end{definition}

A hypergraph is \(\alpha\)-acyclic if and only if it admits a join tree
\cite{beeri1983desirability,fagin1983degrees,brandstadt1998dually}. The presence of a join tree makes \(\alpha\)-acyclic hypergraphs a natural
setting in which to ask whether higher-order hypergraph metrics exhibit
different Euclidean geometry than arbitrary graph metrics. A first possibility
is that the join tree itself approximates the native \(s\)-line metric.
Sections~\ref{sec:join-tree-comparison} and
\ref{sec:canonical-limitations} look at this possibility.

\subsection{Traces and basic join-tree geometry}
\label{subsec:traces}

For an \(s\)-set \(S\in\binom{V}{s}\), define its trace
\[
    \Tr_H(S):=\{e\in E:S\subseteq e\}.
\]
When \(H\) is clear from context, we write simply \(\Tr(S)\).

\begin{lemma}[Trace geometry]
\label{lem:trace-geometry}
Let \(H\) be \(\alpha\)-acyclic and let \(J\) be a join tree for \(H\).
For every \(S\in\binom{V}{s}\):
\begin{enumerate}[label=(\roman*)]
    \item if \(\Tr(S)\neq\varnothing\), then \(\Tr(S)\) induces a connected
    subtree of \(J\);
    \item \(\Tr(S)\) induces a clique in \(\Ls(H)\);
    \item every nonempty trace is contained in a unique
    connected component of \(\Ls(H)\).
\end{enumerate}
\end{lemma}

\begin{proof}
For each \(x\in S\), the set
\[
    E_x=\{e\in E:x\in e\}
\]
induces a connected subtree of \(J\). If \(e,f\in\Tr(S)\), then both contain
every \(x\in S\). Hence the unique \(J\)-path from \(e\) to \(f\) lies in
every \(E_x\), and therefore every node on that path contains all of \(S\).
This proves (i).

If \(e,f\in\Tr(S)\) are distinct, then
\[
    S\subseteq e\cap f,
\]
so \(|e\cap f|\ge s\), and \(e\) and \(f\) are adjacent in \(\Ls(H)\).
Thus \(\Tr(S)\) is a clique, proving (ii). Part (iii) follows immediately.
\end{proof}

So, a trace is connected in the join tree and forms a clique in the native graph. This relationship will be used to compare the two metrics in
\cref{sec:join-tree-comparison}.

\section{Join Trees and the Native \texorpdfstring{$s$}{s}-Line Metric}
\label{sec:join-tree-comparison}

Since \(\alpha\)-acyclic hypergraphs come with a tree representation,
our first question is whether the join-tree metric itself approximates
the native metric.  In this section we show that, on every nontrivial
\(s\)-line component, the join tree restricts to a spanning tree of the native
graph. Further, its exact multiplicative stretch is determined by the
largest join-tree diameter of an \(s\)-trace.

Let \(H=(V,E)\) be \(\alpha\)-acyclic, fix a join tree \(J\), and let \(C\)
be a nontrivial connected component of \(\Ls(H)\).

\subsection{Components induce join subtrees}

\begin{lemma}[Components induce join subtrees]
\label{lem:component-subtree}
The induced subgraph \(J[C]\) is connected.  Consequently, \(J[C]\) is a
subtree of \(J\), and for every \(e,f\in C\), the unique \(J\)-path from
\(e\) to \(f\) is contained in \(C\).
\end{lemma}

\begin{proof}
Fix \(e,f\in C\), and choose a path
\[
  e=e_0,e_1,\ldots,e_k=f
\]
in \(G_C\).  For each \(i\), because \(|e_i\cap e_{i+1}|\ge s\), choose
\(S_i\in\binom{e_i\cap e_{i+1}}{s}\).  Then
\(e_i,e_{i+1}\in\Tr(S_i)\).  By Lemma~\ref{lem:trace-geometry}, each
\(\Tr(S_i)\) is connected in \(J\) and lies entirely in \(C\).  Hence the
unique \(J\)-path from \(e_i\) to \(e_{i+1}\) lies in \(C\).  Concatenating
these paths gives a walk in \(J[C]\) from \(e\) to \(f\), proving that
\(J[C]\) is connected.

Since \(J[C]\) is a connected subgraph of a tree, it is itself a tree.  The
last statement follows from uniqueness of paths in \(J\).
\end{proof}

The previous lemma also implies that the subhypergraph obtained by keeping
only the hyperedges in \(C\) is again \(\alpha\)-acyclic and the tree \(J[C]\)
is a join tree for that restriction.

\begin{lemma}[Join-tree edges inside a native component are native edges]
\label{lem:join-edge-native}
If \(e,f\in C\) and \(ef\in E(J)\), then \(|e\cap f|\ge s\).  Equivalently,
\[
  E(J[C])\subseteq E(G_C).
\]
\end{lemma}

\begin{proof}
Remove the join-tree edge \(ef\).  This separates \(J\) into two components,
say \(A\) containing \(e\) and \(B\) containing \(f\).  Since \(e\) and
\(f\) belong to the same component \(C\) of \(\Ls(H)\), there is a path
\[
  e=e_0,e_1,\ldots,e_k=f
\]
in \(G_C\).  This path must cross the cut \((A,B)\), so for some \(i\) we
have \(e_i\in A\) and \(e_{i+1}\in B\).  Choose
\(S\in\binom{e_i\cap e_{i+1}}{s}\).  The trace \(\Tr(S)\) is connected in
\(J\) by Lemma~\ref{lem:trace-geometry} and contains nodes on both sides of the cut.
Hence it contains the unique join-tree edge crossing the cut, namely \(ef\).
Thus \(S\subseteq e\cap f\), and therefore \(|e\cap f|\ge s\).
\end{proof}

\begin{proposition}[The induced join tree is a spanning tree]
\label{prop:join-tree-spanning}
For every nontrivial component \(C\) of \(\Ls(H)\), the tree \(J[C]\) is a
spanning tree of \(G_C\).
\end{proposition}

\begin{proof}
By Lemma~\ref{lem:component-subtree}, \(J[C]\) is a tree on vertex set \(C\).
By Lemma~\ref{lem:join-edge-native}, every edge of \(J[C]\) is an edge of \(G_C\).
\end{proof}

\subsection{Exact stretch through traces}

For a connected graph \(G\) and a spanning tree \(T\subseteq G\), define the
multiplicative stretch
\[
  \tau(T,G):=\max_{xy\in E(G)}d_T(x,y).
\]
The standard tree-spanner argument gives
\[
  d_G(x,y)\le d_T(x,y)\le \tau(T,G)d_G(x,y)
\]
for all \(x,y\in V(G)\): the first inequality follows because every tree path
is a graph path, and the second by replacing each edge of a shortest
\(G\)-path with its path in \(T\).

For a fixed join tree \(J\), define
\[
  D_s(C,J)
  :=
  \max_{\substack{S\in\binom{V}{s}\\ |\Tr(S)|\ge 2\\ \Tr(S)\subseteq C}}
  \diam_J(\Tr(S)),
\]
where
\[
  \diam_J(\Tr(S)):=\max_{e,f\in\Tr(S)}d_J(e,f).
\]
Because \(C\) is nontrivial, it contains a native edge and thus some trace
with at least two nodes, so \(D_s(C,J)\ge 1\).

\begin{theorem}[Exact trace-diameter/stretch identity]
\label{thm:stretch-identity}
For every nontrivial component \(C\) of \(\Ls(H)\),
\[
  \tau(J[C],G_C)=D_s(C,J).
\]
So, for all \(e,f\in C\),
\[
  d_s(e,f)
  \le d_J(e,f)
  \le D_s(C,J)d_s(e,f).
\]
\end{theorem}

\begin{proof}
Let \(ef\in E(G_C)\).  Choose \(S\in\binom{e\cap f}{s}\).  Then
\(e,f\in\Tr(S)\), and Lemma~\ref{lem:component-subtree} gives
\(d_{J[C]}(e,f)=d_J(e,f)\).  Therefore
\[
  d_{J[C]}(e,f)
  \le \diam_J(\Tr(S))
  \le D_s(C,J).
\]
Taking the maximum over \(ef\in E(G_C)\) yields
\(\tau(J[C],G_C)\le D_s(C,J)\).

Conversely, fix \(S\) with \(|\Tr(S)|\ge 2\) and \(\Tr(S)\subseteq C\).  Any two distinct
\(e,f\in\Tr(S)\) are adjacent in \(G_C\) by
Lemma~\ref{lem:trace-geometry}.  Hence
\[
  d_J(e,f)=d_{J[C]}(e,f)\le \tau(J[C],G_C).
\]
Taking the maximum over \(e,f\in\Tr(S)\), and then over \(S\), gives
\(D_s(C,J)\le\tau(J[C],G_C)\).  This proves the identity. Then, the metric
comparison is the standard spanning-tree inequality applied to \(J[C]\subseteq G_C\).
\end{proof}

\begin{remark}
\label{rem:trace-meaning}
The parameter \(D_s(C,J)\) has a direct geometric interpretation: each nontrivial trace
has native diameter one, but may have large diameter in the join tree.  Thus
\(D_s(C,J)\) measures how severely the join tree \(J\) can stretch a native shortcut.
\end{remark}

\section{Limitations of the Canonical Join Tree}
\label{sec:canonical-limitations}

The comparison in \cref{thm:stretch-identity} is exact for a fixed join tree,
but it may not be useful for Euclidean embedding. A long trace may be highly
stretched in every valid join tree even though that same trace is a clique in
the native graph. The following hypergraph construction displays this behavior.

\begin{proposition}[A forced join path with an equilateral native metric]
\label{prop:forced-path}
For every \(s\ge1\) and \(n\ge3\), there exists an \(\alpha\)-acyclic
hypergraph \(H\) with hyperedges \(e_1,\ldots,e_n\) such that:
\begin{enumerate}[label=(\roman*)]
  \item \(\Ls(H)=K_n\);
  \item every valid join tree is the path
  \[
    e_1-e_2-\cdots-e_n;
  \]
  \item for that join tree \(J\),
  \[
    D_s(E,J)=n-1;
  \]
  \item
  \[
    c_2(E,d_s)=1.
  \]
\end{enumerate}
\end{proposition}

\begin{proof}
Let
\[
  X=\{x_1,\ldots,x_s\}
\]
and introduce distinct vertices \(p_1,\ldots,p_{n-1}\) outside \(X\).  Define
\[
  e_1=X\cup\{p_1\},
\]
\[
  e_i=X\cup\{p_{i-1},p_i\},\qquad 2\le i\le n-1,
\]
and
\[
  e_n=X\cup\{p_{n-1}\}.
\]
The displayed path \(e_1-e_2-\cdots-e_n\) is a join tree: every \(x_j\in X\)
occurs in all hyperedges, while \(p_i\) occurs exactly in the adjacent pair
\(e_i,e_{i+1}\).  So \(H\) is \(\alpha\)-acyclic.

Every pair of hyperedges contains \(X\), so every pair intersects in at least
\(s\) vertices.  Thus \(\Ls(H)=K_n\), and the native metric is equilateral
\[
  d_s(e_i,e_j)=1\qquad(i\neq j).
\]

Now consider any join tree for \(H\).  Since \(p_i\) appears in exactly
\(e_i\) and \(e_{i+1}\), the running-intersection property forces these two
hyperedges to be adjacent in every join tree.  The \(n-1\) forced edges
\(e_ie_{i+1}\) already form a spanning tree, so the join tree is uniquely the
displayed path.

Finally, \(\Tr(X)=E\), whose diameter in the forced path is \(n-1\). Thus
\(D_s(E,J)=n-1\).  On the other hand, every finite equilateral metric embeds
isometrically into Euclidean space as the vertices of a regular simplex, so
\(c_2(E,d_s)=1\).
\end{proof}

Thus, the hypergraph is acyclic and has a canonical join-tree representation, however the representation can be an arbitrarily poor metric. Specifically, even minimizing \(D_s(C,J)\) over all valid join trees cannot repair this construction, because its join tree is forced.

\section{Higher-Order Line Graphs and Dually Chordal Graphs}
\label{sec:dually-chordal}

Dually chordal graphs admit several equivalent descriptions, including maximum
neighborhood orderings and hypergraph representations included in \cite{brandstadt1998dually}. The characterization we use is that a graph is dually chordal if and only if it is the ordinary line graph of an \(\alpha\)-acyclic hypergraph \cite[Theorem~9(ii)]{brandstadt1998dually}.  We show that the same graph class arises at every threshold \(s\ge1\).

\begin{theorem}[Characterization of \(s\)-line graphs]
\label{thm:s-line-characterization}
Fix \(s\ge1\).  A finite graph \(G\) is dually chordal if and only if there
exists an \(\alpha\)-acyclic hypergraph \(H\) such that
\[
  G\cong \Ls(H).
\]
Consequently,
\[
  \{\Ls(H):H\text{ is }\alpha\text{-acyclic}\}
  =
  \{G:G\text{ is dually chordal}\}.
\]
\end{theorem}

\begin{proof}
We prove both directions.

\smallskip
\noindent\emph{Forward direction.}
Let \(H=(V,E)\) be \(\alpha\)-acyclic with join tree \(J\).  Construct a new
hypergraph
\[
  \widehat H_s=(\widehat V,\widehat E)
\]
whose ground set contains one vertex for every \(s\)-subset of \(V\) lying in
at least one hyperedge, together with a private vertex \(p_e\) for each
\(e\in E\):
\[
  \widehat V
  =
  \left\{S\in\binom{V}{s}: S\subseteq e\text{ for some }e\in E\right\}
  \cup \{p_e:e\in E\}.
\]
For each \(e\in E\), define
\[
  \widehat e
  :=
  \binom{e}{s}\cup\{p_e\}.
\]
The private vertices keep the hyperedges \(\widehat e\) distinct, including
when \(|e|<s\).

For distinct \(e,f\in E\), the private vertices do not contribute to
\(\widehat e\cap\widehat f\), and therefore
\[
  \widehat e\cap\widehat f\neq\varnothing
  \iff
  \exists S\in\binom{V}{s}\text{ with }S\subseteq e\cap f
  \iff
  |e\cap f|\ge s.
\]
So
\[
  L(\widehat H_s)\cong \Ls(H).
\]

It remains to show that \(\widehat H_s\) is \(\alpha\)-acyclic. Identify the
node \(\widehat e\) with \(e\) and use the original tree \(J\). For a vertex
\(S\in\widehat V\) corresponding to an \(s\)-set, the hyperedges of
\(\widehat H_s\) containing \(S\) are
\[
  \{\widehat e:S\subseteq e\},
\]
which correspond to the trace \(\Tr_H(S)\).  This set is connected in \(J\) by
Lemma~\ref{lem:trace-geometry}. Each private vertex \(p_e\) occurs in one hyperedge
only. Thus \(J\) is a join tree for \(\widehat H_s\). By the ordinary
line-graph characterization of \cite[Theorem~9(ii)]{brandstadt1998dually}, \(L(\widehat H_s)\), and thus \(\Ls(H)\), is dually chordal.

\smallskip
\noindent\emph{Reverse direction.}
Suppose \(G\) is dually chordal.  By
\cite[Theorem~9(ii)]{brandstadt1998dually}, there exists an
\(\alpha\)-acyclic hypergraph \(K=(X,\cE)\) such that
\[
  G\cong L(K).
\]
For each \(x\in X\), create a block
\[
  B_x:=\{x_1,\ldots,x_s\}
\]
of \(s\) fresh vertices, with the blocks pairwise disjoint.  Replace each
\(A\in\cE\) by
\[
  A^{(s)}:=\bigcup_{x\in A}B_x,
\]
and let \(K^{(s)}\) denote the resulting hypergraph. Since the blocks
are nonempty and pairwise disjoint, the map \(A\mapsto A^{(s)}\) is
injective, so the resulting hyperedges are distinct.

For \(A,B\in\cE\),
\[
  A^{(s)}\cap B^{(s)}
  =
  \bigcup_{x\in A\cap B}B_x.
\]
Because the blocks are pairwise disjoint and each has size \(s\),
\[
  |A^{(s)}\cap B^{(s)}|\ge s
  \iff
  A\cap B\neq\varnothing.
\]
Therefore
\[
  \Ls(K^{(s)})\cong L(K)\cong G.
\]
Finally, any join tree for \(K\) is also a join tree for \(K^{(s)}\), because
all copies of a vertex \(x\) occur in exactly the hyperedges that originally
contained \(x\).  Hence \(K^{(s)}\) is \(\alpha\)-acyclic.
\end{proof}

\section{Tree Approximation and the Euclidean Upper Bound}
\label{sec:euclidean-upper}

We know that every connected dually chordal graph has a spanning tree that is also an additive \(3\)-spanner and a multiplicative \(4\)-spanner \cite{brandstadt1999distance}. Specifically, if \(G\) is connected and dually chordal, there exists a spanning tree \(T\subseteq G\) satisfying
\[
  d_G(x,y)\le d_T(x,y)\le 4d_G(x,y)
\]
for all \(x,y\in V(G)\).

Combining this theorem with \cref{thm:s-line-characterization} removes the
trace-diameter factor from \cref{thm:stretch-identity}.

\begin{theorem}[Universal tree approximation of the native metric]
\label{thm:native-tree-spanner}
Let \(H\) be \(\alpha\)-acyclic, fix \(s\ge1\), and let \(C\) be a connected
component of \(\Ls(H)\).  There exists a spanning tree \(T^\star\) of \(G_C\)
such that, for all \(e,f\in C\),
\[
  d_s(e,f)
  \le d_{T^\star}(e,f)
  \le 4d_s(e,f).
\]
\end{theorem}

\begin{proof}
By Lemma~\ref{lem:component-subtree}, the restriction of \(H\) to the hyperedges
in \(C\) is \(\alpha\)-acyclic.  Applying \cref{thm:s-line-characterization} to this restricted hypergraph shows that \(G_C\) is dually chordal. The multiplicative \(4\)-spanner theorem \cite{brandstadt1999distance} then gives a spanning tree \(T^\star\subseteq G_C\) satisfying
\[
  d_{T^\star}(e,f)\le4d_s(e,f).
\]
Since every edge of \(T^\star\) is an edge of \(G_C\), every tree path is a
native graph path, so \(d_s(e,f)\le d_{T^\star}(e,f)\).
\end{proof}

We now pass from tree approximation to Euclidean embedding.  Matou\v{s}ek
proved that every \(n\)-point tree metric embeds into uniformly convex Banach
spaces with the asymptotically optimal dependence required here; in particular,
for Hilbert space,
\[
  c_2(T)=O\!\left(\sqrt{\log\log n}\right)
\]
for every \(n\)-vertex tree \(T\) \cite{matousek1999embedding}.  This is the
tree-embedding input used below.

\begin{proposition}[Euclidean distortion of dually chordal graphs]
\label{prop:dually-chordal-c2}
Every connected dually chordal graph \(G\) on \(n\ge3\) vertices satisfies
\[
  c_2(G)
  =O\!\left(\sqrt{\log\log n}\right).
\]
\end{proposition}

\begin{proof}
Let \(T\subseteq G\) be the multiplicative \(4\)-spanning tree from
\cite{brandstadt1999distance}.  By \cite{matousek1999embedding}, there is an
embedding \(F:V(T)\to\elltwo\) of distortion
\(D=O(\sqrt{\log\log n})\).  Rescale \(F\) so that
\[
  \frac{1}{D}d_T(x,y)
  \le \|F(x)-F(y)\|_2
  \le d_T(x,y).
\]
Using \(d_G\le d_T\le4d_G\), we obtain
\[
  \frac{1}{D}d_G(x,y)
  \le \|F(x)-F(y)\|_2
  \le 4d_G(x,y).
\]
Thus \(F\), viewed as an embedding of \((V(G),d_G)\), has distortion at most
\(4D=O(\sqrt{\log\log n})\).
\end{proof}

\begin{theorem}[Euclidean upper bound for acyclic \(s\)-line metrics]
\label{thm:hypergraph-upper}
Let \(H\) be \(\alpha\)-acyclic, fix \(s\ge1\), and let \(C\) be a connected
component of \(\Ls(H)\) with \(|C|=n\ge3\).  Then
\[
  c_2(C,d_s)
  =O\!\left(\sqrt{\log\log n}\right),
\]
where the implicit constant is universal and, in particular, independent of
\(s\), \(H\), and the choice of join tree.
\end{theorem}

\begin{proof}
By \cref{thm:s-line-characterization} applied componentwise as in
\cref{thm:native-tree-spanner}, the graph \(G_C\) is dually chordal.  The
claim follows from \cref{prop:dually-chordal-c2}.
\end{proof}

The upper bound is stronger than what we obtained by embedding the join tree as there is no parameter analogous to \(D_s(C,J)\).  The next section shows that the \(\sqrt{\log\log n}\) upper bound is tight.

\section{Tightness and Bounded-Rank Realizations}
\label{sec:tightness}

The upper bound in \cref{thm:hypergraph-upper} is obtained by comparing the
native metric to a tree metric.  We now show that every finite tree metric occurs as the native \(s\)-line metric of an \(\alpha\)-acyclic hypergraph.  This transfers the classical lower bound for complete binary trees directly into the hypergraph setting.

\subsection{Realizing arbitrary tree metrics}

\begin{theorem}[Tree realization]
\label{thm:tree-realization}
Let \(Q=(U,F)\) be any finite tree and fix \(s\ge1\).  There exists an
\(\alpha\)-acyclic hypergraph \(H_{Q,s}\) with one hyperedge \(E_u\) for each
\(u\in U\) such that
\[
  \Ls(H_{Q,s})\cong Q.
\]
Moreover,
\[
  |E_u|=1+s\deg_Q(u)
\]
for every \(u\in U\), and thus
\[
  \rank(H_{Q,s})\le 1+s\Delta(Q).
\]
\end{theorem}

\begin{proof}
For every tree edge \(uv\in F\), create a fresh block \(S_{uv}\) of \(s\) vertices, with all created blocks pairwise disjoint. For each tree vertex
\(u\in U\), also create a private vertex \(p_u\). Define
\[
  E_u
  :=
  \{p_u\}
  \cup
  \bigcup_{uv\in F}S_{uv}.
\]

If \(u\) and \(v\) are adjacent in \(Q\), then the only shared block is
\(S_{uv}\), so
\[
  E_u\cap E_v=S_{uv}
  \qquad\text{and hence}\qquad
  |E_u\cap E_v|=s.
\]
If \(u\) and \(v\) are not adjacent, no tree edge is incident to both
vertices, and the blocks and private vertices are pairwise disjoint. Thus
\[
  E_u\cap E_v=\varnothing.
\]
Therefore
\[
  E_uE_v\in E(\Ls(H_{Q,s}))
  \iff
  uv\in F,
\]
so \(\Ls(H_{Q,s})\cong Q\).

To verify \(\alpha\)-acyclicity, use \(Q\) as a tree on the hyperedges
\(\{E_u:u\in U\}\).  Every ground vertex in a block \(S_{uv}\) belongs only to \(E_u\) and \(E_v\), which are adjacent in \(Q\); each private vertex
\(p_u\) belongs only to \(E_u\).  Thus the running intersection property holds,
so \(Q\) is a join tree for \(H_{Q,s}\).

Lastly, \(E_u\) contains its one private vertex with one disjoint
\(s\)-block for each edge incident to \(u\), giving
\[
  |E_u|=1+s\deg_Q(u).
\]
The rank bound follows immediately.
\end{proof}

\subsection{Matching lower bound}

Bourgain's work on metric characterizations of superreflexivity implies the
classical lower bound for complete binary trees in Hilbert space
\cite{bourgain1986superreflexivity}; a proof was later given by
Linial and Saks \cite{linial2003binary}. If \(B_h\) denotes the complete
binary tree of depth \(h\), then
\[
  c_2(B_h)=\Omega(\sqrt{\log h}).
\]
Since \(|V(B_h)|=2^{h+1}-1\), this is equivalently
\[
  c_2(B_h)
  =\Omega\!\left(\sqrt{\log\log |V(B_h)|}\right).
\]

For \(n\ge1\), define the worst-case distortion
\[
  \Ds(n)
  :=
  \sup\left\{
    c_2(C,d_s):
    \begin{array}{l}
      H\text{ is }\alpha\text{-acyclic},\\
      C\text{ is a connected component of }\Ls(H),\\
      |C|=n
    \end{array}
  \right\}.
\]

\begin{theorem}[Tight worst-case Euclidean distortion]
\label{thm:tight-distortion}
For every fixed \(s\ge1\), as \(n\to\infty\),
\[
  \Ds(n)
  =
  \Theta\!\left(\sqrt{\log\log n}\right).
\]
\end{theorem}

\begin{proof}
The upper bound
\[
  \Ds(n)=O\!\left(\sqrt{\log\log n}\right)
\]
follows from \cref{thm:hypergraph-upper}.

For the lower bound, first consider \(n=N_h:=2^{h+1}-1\).  By
\cref{thm:tree-realization}, there exists an \(\alpha\)-acyclic hypergraph
\(H_{B_h,s}\) for which
\[
  \Ls(H_{B_h,s})\cong B_h.
\]
Therefore the binary-tree lower bound
\cite{bourgain1986superreflexivity,linial2003binary} gives
\[
  \Ds(N_h)
  \ge c_2(B_h)
  =\Omega\!\left(\sqrt{\log\log N_h}\right).
\]

For an arbitrary sufficiently large \(n\), choose the largest \(h\) with
\(N_h\le n\).  Then \(N_h=\Theta(n)\).  Form an \(n\)-vertex tree \(Q_n\) by
attaching a path of length \(n-N_h\) to a leaf of \(B_h\).  The original
\(B_h\) remains an isometric subspace of \(Q_n\): because \(Q_n\) is a tree,
the unique path between two vertices of \(B_h\) is unchanged by attaching the
new pendant path.  Then, restricting any embedding of \(Q_n\) to this isometric
subspace gives
\[
  c_2(Q_n)\ge c_2(B_h)
  =\Omega\!\left(\sqrt{\log\log n}\right).
\]
Applying \cref{thm:tree-realization} to \(Q_n\) realizes this metric as an
\(s\)-line metric of an \(\alpha\)-acyclic hypergraph.  So
\[
  \Ds(n)=\Omega\!\left(\sqrt{\log\log n}\right),
\]
which matches the upper bound.
\end{proof}

\subsection{The lower bound already holds at bounded rank}

The construction above shows that the obstruction does not depend on
hyperedges whose size grows with \(n\).

\begin{corollary}[Bounded-rank lower bound]
\label{cor:bounded-rank-lower}
Fix \(s\ge1\).  For every sufficiently large \(n\), there exists an
\(\alpha\)-acyclic hypergraph \(H\) of rank at most
\[
  3s+1
\]
and a connected component \(C\) of \(\Ls(H)\) with \(|C|=n\) such that
\[
  c_2(C,d_s)
  =\Omega\!\left(\sqrt{\log\log n}\right).
\]
So, imposing the rank bound \(\rank(H)\le3s+1\) does not improve the
worst-case asymptotic Euclidean distortion for fixed \(s\).
\end{corollary}

\begin{proof}
Use the tree \(Q_n\) from the proof of \cref{thm:tight-distortion}.  A complete
binary tree has maximum degree at most \(3\), and attaching a path to one of
its leaves preserves the maximum-degree bound \(\Delta(Q_n)\le3\).  By
\cref{thm:tree-realization}, the corresponding hypergraph satisfies
\[
  \rank(H_{Q_n,s})
  \le 1+s\Delta(Q_n)
  \le 3s+1.
\]
Its native \(s\)-line metric is exactly \(Q_n\), whose Euclidean distortion is
\(\Omega(\sqrt{\log\log n})\) by the proof of
\cref{thm:tight-distortion}.
\end{proof}

Thus, \cref{thm:hypergraph-upper,thm:tight-distortion} identify the
correct worst-case order of Euclidean distortion for native \(s\)-line metrics
of \(\alpha\)-acyclic hypergraphs. The bounded-rank construction in
Corollary~\ref{cor:bounded-rank-lower} shows that the lower bound still holds when the hyperedges have size bounded as a function of the threshold \(s\).

\section{Concluding Remarks}

We have determined the worst-case Euclidean distortion of native \(s\)-line metrics arising from \(\alpha\)-acyclic hypergraphs. For every fixed \(s\geq 1\), the asymptotic order is
\[
\Theta\!\left(\sqrt{\log\log n}\right).
\]
This separates the acyclic setting from general hypergraphs, whose native \(s\)-line metrics recover arbitrary finite graph metrics and have worst-case distortion \(\Theta(\log n)\).

We note that the improved geometry does not follow directly from the canonical tree representation of an \(\alpha\)-acyclic hypergraph. The relevant structure instead appears after passing to the \(s\)-line graph itself where, for every threshold \(s\), the \(s\)-line graphs of \(\alpha\)-acyclic hypergraphs are exactly the dually chordal graphs. Their constant-factor spanning-tree approximations then allow the Euclidean geometry of tree metrics to be transferred to the hypergraph setting. Understanding when these two notions of tree-likeness coincide, and when they diverge, seems to be a useful direction for the broader study of hypergraph metrics. The matching lower bound comes from our work that every finite tree can be realized exactly as the native \(s\)-line metric of an \(\alpha\)-acyclic hypergraph. Thus, the known obstruction given by complete binary trees already occurs inside this hypergraph class. Further, the realization can be carried out with rank at most \(3s+1\), showing that the lower bound is not caused by hyperedges whose sizes grow with the number of hyperedges.

However, what happens under other notions of hypergraph acyclicity remains open. Determining which other acyclic or approximately acyclic classes admit improved Euclidean embedding bounds may reveal a broader connection between structural decompositions of hypergraphs and the geometry of their higher-order line metrics. A more general question is whether there exists a class of hypergraphs whose Euclidean embedding does not match that of its general graph analogue. 

The characterization by dually chordal graphs also suggests more refined metric questions. Our argument uses the existence of a multiplicative \(4\)-spanning tree, which is sufficient for the asymptotic Euclidean bound but does not attempt to optimize the tree approximation itself. It would be interesting to determine which \(s\)-line graphs arising from \(\alpha\)-acyclic hypergraphs admit better multiplicative or additive tree approximations.

In addition, since the main reduction passes through tree metrics,
analogous questions can be posed for embeddings into other Banach spaces
or into low-dimensional Euclidean spaces, as well as for algorithmic
problems whose guarantees depend on tree-like metric structure.

\section*{Acknowledgments}
I thank Fan Wei for her guidance and encouragement throughout this project,
and for helpful feedback on the manuscript.

\end{document}